\documentclass[12pt]{article}

\usepackage{latexsym,amsfonts,amssymb,amsmath,amsthm,graphicx,cite}

\newtheorem{theorem}{Theorem}
\newtheorem{lemma}{Lemma}
\newtheorem{remark}{Remark}

\newcommand{\Cc}{\mathcal{C}}

\DeclareMathOperator{\erf}{erf}

\title{Iterative method for solving nonlinear integral
equations describing rolling solutions in String Theory}
\author{L.~Joukovskaya\\
Steklov Mathematical Institute\\
119991, Moscow, Gubkin St. 8}

\date{}

\begin{document}
\maketitle
\begin{abstract}
We consider a nonlinear integral equation with infinitely many
derivatives that appears when a system of interacting open and closed
strings is investigated if the nonlocality in the closed string sector
is neglected. We investigate the properties of this equation,
construct an iterative method for solving it, and prove that the
method converges.
\end{abstract}

Nonlinear equations with infinite number of time derivatives
have recently became the subject of investigation in both
standard and p-adic string theories
\cite{Sen,MolZw,Y,AJ,MS,BFOW,FO,BF,VV,VSV}. Properties of
some of them were systematically studied in \cite{VV}.
In present work further investigations of such equations will be
presented, we will consider equation of the following form which
appears in string theory
\begin{equation}
\label{eq1}
a\Phi^3(t)+(1-a)\Phi(t) = \exp\big(a\frac{d^2}{dt^2}\big)\Phi(t),
\end{equation}
where $a$ is a constant, $a\in(0,1]$.

Precise meaning to equation (\ref{eq1}) will be given below.
We note that in the case $a=1$, this equation takes a form
of the equation for $p$-adic string with $p=3$, which was considered
in the papers \cite{BFOW,Y,VV}.

Equation  (\ref{eq1}) arises when studying system of
interacting open and closed string in the approximation when
nonlocality in the interaction of closed string is neglected
\cite{Oh,LY,L1}. Note that resulting equations are still nonlocal
in terms of open string tachyon.
It is interesting to note that in the mechanical
approximation \cite{AJ,LY} this equation transforms to equation
$$
\Phi^3(t)-\Phi(t)=\frac{d^2}{dt^2}\Phi(t),
$$
which have well known kink solution
$$
\Phi(t)=\tanh(\frac{t}{\sqrt{2}}).
$$
Note that kink usually describes the solution depending on spatial coordinates.

In this paper  we will investigate the properties of equation
(\ref{eq1}) and consider boundary value problems for bounded
solution. In particular, we will construct rolling solutions
interpolating between two vacua.

Equation (\ref{eq1}) is a pseudo-differential equation with
symbol $e^{-a \xi^2}$ which for positive $a$ might be presented
as a nonlinear integral equation (in full accordance with the
equations considered in \cite{MolZw,VV,VSV})

\begin{equation}
\label{IE}
{C}_a[\Phi](t)=a\Phi^3(t)+(1-a)\Phi(t),
\end{equation}
where constant $a$:  $0< a < 1$ and operator ${C}_a$ is defined as
\begin{equation}
\label{ifou}
C_a[\psi] (t)=\int
_{-\infty}^{\infty} \Cc_{a}[(t-\tau)^{2}] \psi(\tau)d\tau,
\end{equation}
with the kernel
$$
\Cc_{a}[(t-\tau)^{2}]=\frac{1}{\sqrt{4\pi a}}e^{-\frac{(t-\tau)^{2}}{4a}}.
$$
We look for solutions of equation (\ref{IE}) in the class of
real-valued measurable functions.

\begin{theorem}
If solution $\Phi(t)$ of equation (\ref{IE}) is
bounded, then it satisfies the estimate
\begin{equation}
|\Phi(t)|\leq 1,~~ t\in \mathbb{R}
\end{equation}
\end{theorem}
\begin{proof}
Let us suppose that solution $\Phi(t)$ of equation (\ref{IE})
is bounded, i.e. there exists such a number $M>0$ that
\begin{equation}
\sup_{t}|\Phi(t)|=M.
\end{equation}
It follows from(\ref{IE}) and (\ref{ifou}) that
\begin{equation}
|a\Phi^3(t)+(1-a)\Phi(t)|=|\int^{+ \infty}_{- \infty} \Phi(\tau) \Cc_{a}[(t-\tau)^2]d\tau|\leq
\end{equation}
$$
\leq
\int^{+ \infty}_{- \infty} |\Phi(\tau)| \Cc_{a}[(t-\tau)^2]d\tau\leq
\sup_{\tau}|\Phi(\tau)|\int^{+ \infty}_{- \infty}\Cc_{a}[(t-\tau)^2]d\tau=M,
$$
and
$$
\sup_{t}|a\Phi^3(t)+(1-a)\Phi(t)|=aM^3+(1-a)M\leq M,
$$
i.e.
$M \leq 1$. This proves the theorem.
\end{proof}

\begin{remark}
\label{remark1}
The theorem 1, as well as theorems 4 and 5 proved below,
are similar to the corresponding theorems from papers \cite{MolZw, VV},
but are proved here using features of more general equation (\ref{IE}).
\end{remark}
\begin{lemma}
\label{lemma1}
If function $\Phi(t)$ is bounded, then function $C_{a}[\Phi](t)$ is continuous in $t$.
\end{lemma}
\begin{proof}
By the formulation of lemma the function $\Phi(t)$ is bounded,
i.e. there exists such a number $M>0$, for which
$\sup\limits_{t}|\Phi(t)|=M$. Let us consider the estimate
$$
\left|C_{a}[\Phi](t+\delta)-C_{a}[\Phi](t)\right|=
\frac{1}{\sqrt{4\pi a}}\left|\int_{-\infty}^{+\infty}
(e^{-\frac{((t+\delta)-\tau)^2}{4a}}-e^{-\frac{(t-\tau)^2}{4a}})\Phi(\tau)d\tau\right|\leqslant
$$
$$
\leqslant
\frac{1}{\sqrt{4\pi a}}\int_{-\infty}^{+\infty}
\left|e^{-\frac{((t+\delta)-\tau)^2}{4a}}-e^{-\frac{(t-\tau)^2}{4a}}\right|\cdot|\Phi(\tau)|d\tau
\leqslant
$$
$$
\leqslant\sup_{\tau}\Phi(\tau)\frac{1}{\sqrt{4 \pi a}}\int_{-\infty}^{+\infty}
\left|e^{-\frac{((t+\delta)-\tau)^2}{4a}}-e^{-\frac{(t-\tau)^2}{4a}}\right|d\tau\leqslant
$$
$$
=\frac{M}{\sqrt{4 \pi a}}\int_{-\infty}^{+\infty}
\left|e^{-\frac{(y+\delta)^2}{4a}}-e^{-\frac{y^2}{4a}}\right|dy
$$
Evaluating the absolute value in the integrand, we get
\begin{equation}
\label{C}
|C_{a}[\Phi](t+\delta)-C_{a}[\Phi](t)|
\leqslant 2M\erf\left(\frac{|\delta|}{4\sqrt{a}}\right)
\end{equation}
where $\erf(t)$ is an error function, which is defined by
\begin{equation}
\label{erf}
\erf(t)=\frac{2}{\sqrt{\pi}}\int_{0}^{t}e^{-\tau^2}d\tau.
\end{equation}
The estimate (\ref{C}) with (\ref{erf}) taken into account
shows that the function $C_{a}[\Phi](t)$ is continuous.
\end{proof}

\begin{theorem}
Any bounded solution $\Phi(t)$ of equation (\ref{IE}) is continuous.
\end{theorem}
\begin{proof}
By Lemma 1 the functional  $C_{a}[\Phi](t)$ is continuous.
Let us prove that the solution $\Phi(t)$ itself is continuous.
Let us suppose the opposite, i.e. that $\Phi(t)$
is not continuous; it then follows that functions
$a\Phi^3(t)$ and  $(1-a)\Phi(t)$ are also not continuous.
Taking into account that $a\in [0,1]$, we get that the function
$a\Phi^3(t)+(1-a)\Phi(t)$ is also not continuous, which contradicts
the statement that function $\Phi(t)$ solves equation (\ref{IE}) since
the function ${\cal C}_a\Phi(t)$ should be continuous.
Thus the continuity of the function $\Phi(t)$ is proved.
\end{proof}

\begin{theorem}
If solution  $\Phi(t)$ of equation (\ref{IE}) is positive and
bounded, then result of acting with operator $C_{a}$ is
decreasing, i.e. $C_{a}[\Phi](t)\leq \Phi(t)$, and if solution
$\Phi(t)$ of the equation (\ref{IE}) is negative and bounded,
then result of acting with operator $C_{a}$ is increasing, i.e.
$C_{a}[\Phi](t)\geq \Phi(t)$.
\end{theorem}
\begin{proof}
Let us prove the first statement of the theorem.
Following assumption that the solution $\Phi(t)$
of the equation (\ref{IE}) is positive and bounded,
by Theorem 1 we get $|\Phi(t)|\leq 1$ and thus
$$
C_{a}[\Phi](t)=a\Phi^3(t)+(1-a)\Phi(t) \leq \Phi(t)
$$
i.e. $C_{a}[\Phi](t)\leq \Phi(t)$ for positive and bounded solutions.

The second statement of the Theorem is proved analogously.
\end{proof}

\begin{theorem}
If solution $\Phi(t)$ of equation (\ref{IE}) has a definite limit as $t\to \infty$,
then the limit takes one of the following possible values: $-1$, $0$, $1$.
\end{theorem}
\begin{proof}
Following the formulation of the Theorem solution of equation (\ref{IE}) has a limit
$\lim\limits_{t\to+\infty}\Phi(t)$, let $\lim\limits_{t\to+\infty}\Phi(t)=b$,
calculating $C_{a}[\Phi(t)]$
$$
\lim_ {t\to+\infty}C_{a}[\Phi(t)]\lim_ {t\to+\infty}\frac{1}{\sqrt{4\pi a}}\int_{-\infty}^{+\infty}
e^{-\frac{(t-\tau)^{2}}{4a}}\Phi(\tau)d\tau
$$
$$
\frac{1}{\sqrt{4\pi a}}\left[\lim_ {t\to+\infty}\int_{-\infty}^{0}
e^{-\frac{(t-\tau)^{2}}{4a}}\Phi(\tau)d\tau+
\lim_ {t\to+\infty}\int_{0}^{+\infty}
e^{-\frac{(t-\tau)^{2}}{4a}}\Phi(\tau)d\tau\right]
$$
and making change of variables $\tau \to t-u$ in the first integral
and $\tau \to t+u$ in the second one we get
$$
\frac{1}{\sqrt{4\pi a}}\left[\lim_ {t\to+\infty}\int_{t}^{+\infty}
e^{-\frac{u^{2}}{4a}}\Phi(t-u)du+
\lim_ {t\to+\infty}\int_{-t}^{+\infty}
e^{-\frac{u^{2}}{4a}}\Phi(t+u)du\right]=$$
 $$
=\frac{1}{\sqrt{4\pi a}} \left[ \int_{+\infty}^{+\infty}
e^{-\frac{u^{2}}{4a}}b du+
\int_{-\infty}^{+\infty}
e^{-\frac{u^{2}}{4a}} b du\right]= \frac{b}{\sqrt{4\pi a}}\int_{-\infty}^{+\infty}
e^{-\frac{u^{2}}{4a}} du,
$$
i.e. we obtain that $\lim\limits_ {t\to+\infty}C_{a}[\Phi(t)]=b$.

Taking the limit $t\to+\infty$ in the equation (\ref{IE}), we
obtain the equation
$$
ab^3+(1-a)b=b,
$$
which has three roots $b=-1$, $b=0$, and $b=1$.

Subsequently we proved that the limit might takes only the
following values $0$, $-1$ è $1$.
\end{proof}

\begin{theorem}
There exists a unique nonnegative bounded continuous solution
$\Phi(t)\equiv 1$ of equation (\ref{IE}) that satisfies the
boundary conditions
\begin{equation}
\label{krz1}
\lim_{t\to\pm\infty} \Phi(t)= 1
\end{equation}
\end{theorem}
\begin{proof}
Note that $\Phi(t)\equiv 1$ is a solution of the boundary value problem
in \ref{IE}-\ref{krz1}).
Let $\Phi^{*}(t),~0 \leq \Phi^{*}(t)\not\equiv 1$ be another bounded solution of
this problem. Then
$0 \leq \Phi^{*}(t) \leq 1$ and by (\ref{krz1}) there exists a $t_{0}$ such that
\begin{equation}
\label{oc}
0 \leq \Phi^{*}(t_0)=\min_{t}\Phi^{*}(t)\leq 1.
\end{equation}
From equation (\ref{IE}) we have
\begin{equation}
\label{ner}
a{\Phi^{*}}^3(t_0)+(1-a)\Phi^{*}(t_0)
=\int^{+ \infty}_{- \infty} \Phi^{*}(\tau) \Cc _{a}[(t_0-\tau)^2]d\tau \geq \Phi^{*}(t_0),
\end{equation}
this inequality holds if $\Phi^{*}(t_0)\geq 1$ and
$\Phi^{*}(t_0)=0$. Taking into consideration that
$|\Phi^{*}(t)| \leq 1$, we obtain that inequality (\ref{ner})
holds only if $\Phi^{*}(t_0)=0$ or $\Phi^{*}(t_0)=1$, but the
solution of inequality (\ref{ner}), $\Phi^{*}(t_0)=0$, does not
satisfy the boundary problem (\ref{krz1}). Hence we consider
$\Phi^{*}(t) \geq 0$ and  $\exists t_0$ :
$C_{a}[\Phi^{*}](t_0)=0$, then $\Phi^{*}(t)\equiv 0$, i.e.
there exists a unique nonnegative solution
$\Phi^{*}(t) \equiv \Phi (t)\equiv1$ of the boundary value problem (\ref{krz1}), (\ref{IE}).
\end{proof}

\begin{theorem}
There exists a continuous solution of equation(\ref{IE})
that satisfies the boundary conditions
\begin{equation}
\label{krz2}
\lim \Phi(t)=
  \begin{cases}
  1, &t\to\infty\\
  -1, &t\to-\infty.
  \end{cases}
\end{equation}
Moreover the iterative procedure
\begin{equation}
\label{IE-IP}
{\cal C}_a\Phi_{n}=a\Phi_{n+1}^3+(1-a)\Phi_{n+1}
\end{equation}
converges to this solution.
\end{theorem}
\begin{proof}
Taking into account the invariance of the equation (2) under
the change of variables $\phi(t) \to -\phi(-t)$, we will seek
for odd solutions. The boundary value problem (2), (12)
rewrites on the positive semiaxis $t\geq 0$ as
\begin{equation}
\label{14}
{K}_a [\phi](t)=a\phi^3(t)+(1-a)\phi(t),~~~~t\geq 0,
\end{equation}
and
\begin{equation}
%\tag{14'}
\label{14z}
\lim_{t\to+\infty} \phi(t)=1
\end{equation}
where the operator ${K}_a $ is given by
$$
{K}_a [\phi](t)=\int_{0}^{\infty}{\cal K}_{a}(t,\tau)\phi(\tau)d\tau,
$$
with
$$
{\cal K}_{a}(t,\tau)=\frac{1}{\sqrt{4 \pi a}}
\left[ e^{-\frac{(t-\tau)^2}{4a}}- e^{-\frac{(t+\tau)^2}{4a}} \right].
$$
Solution of the original problem  $\Phi(t)$ is given by
\begin{equation}
%\tag{14a}
\label{14zz}
\Phi(t)=
\begin{cases}
\phi(t),  &t\geq 0,\\
-\phi(-t),&t<0.
\end{cases}
\end{equation}

We will seek solution of equation (\ref{IE}), using an
iterative procedure (\ref{IE-IP}) which on the semiaxis $t\geq
0$ takes the following form
\begin{equation}
\label{itpr}
{K}_a\phi_{n}=a\phi_{n+1}^3+(1-a)\phi_{n+1}
\end{equation}
Solving the cubic equation \ref{itpr}) for $\phi_{n+1}$, we
obtain
\begin{subequations}
\begin{equation}
\label{phin}
\phi_{n+1}=-v_{n}(1-a)+\frac{1}{3av_{n}},
\end{equation}
where
\begin{equation}
\label{vn}
v_{n}=\left(\frac{2}{27a^2 B_n +\sqrt{108(1-a)^3a^3+729 a^4 B_n^2}}\right)^{1/3},
\end{equation}
and
\begin{equation}
B_n=K_a \phi_n.
\end{equation}
\end{subequations}
As initial iteration we take the function
$$
\phi_0 = \frac{1}{2}(1-e^{-(a t)^2}).
$$
Acting with operator ${K}_a$, we get
$$
{K}_a \phi_0 = \frac{1}{2}(1-e^{-\frac{{(a t)^2}}{\sqrt{1+a^2}}}).
$$
Evaluating $a\phi_{0}^3+(1-a)\phi_{0}$ we obtain
$$
a\phi_{0}^3+(1-a)\phi_{0} \leq {K}_a\phi_0,
$$
i.e.
$$
a\phi_{0}^3+(1-a)\phi_{0} \leq a\phi_{1}^3+(1-a)\phi_{1},
$$
subsequently $\phi_0 \leq \phi_1$.
\begin{figure}
%\centering
\includegraphics[width=5cm]{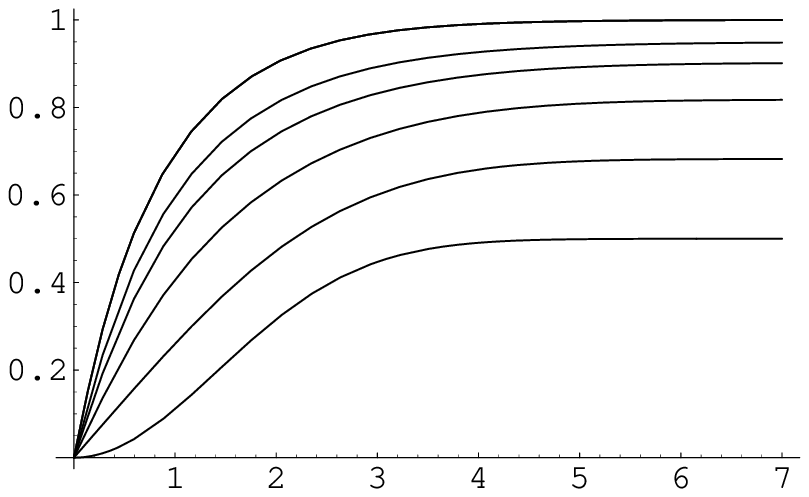}~a)~~~~~~~~~~~~~
\includegraphics[width=5cm]{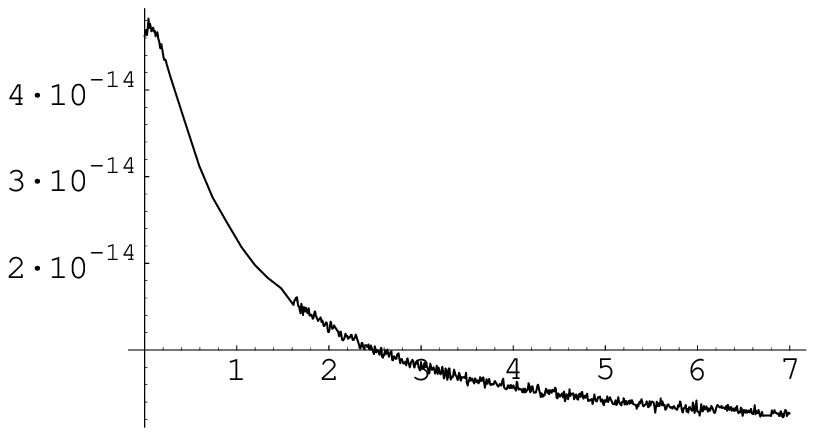}~b)
\caption{a) results of iterative procedure
$\phi_0$, $\phi_1$, $\phi_2$, $\phi_3$, $\phi_4$, $\phi_{50}$, $\phi_{150}$
from below to up correspondingly;
b)--the difference between $\phi_{50}$ and $\phi_{150}$-th iterations.}
\label{iterLY}
\end{figure}
Taking into account the fact that the kernel is nonnegative
after integration we obtain $K_{a}\phi_0 \leq K_{a}\phi_1$,
i.e. $B_0 \leq B_1.$

Using explicit expressions (\ref{vn}) and (\ref{phin}) for the
functions $v_{n}$ and $\phi_{n+1}$ we conclude that inequality
$B_0 \leq B_1$, implies inequalities  $v_0 \geq v_1$, and
$\phi_{1}\leq \phi_{2}$. Repeating the above argument $(n-1)$
times we obtain
$$
\phi_0 \leq \phi_1\leq...\leq \phi_n \leq \phi_{n+1}.
$$
This illustrates the figure of calculations of the iterations
(see fig.\ref{iterLY}).

We now prove that
$$
\phi_{n} \leq 1.
$$
Iterative procedure have the form (\ref{itpr}). We can see that
initial iteration $\phi_0$ is bounded, $\phi_0 < 1$. Then we
have $K_{a} \phi_0 \leq \phi_0<1$, and therefore
$a\phi_{1}^3+(1-a)\phi_{1}<1$. Let us suppose that there exists
$t_1 \geq 0$ such that $\phi_1(t_1)>1$, then
$a\phi_{1}^3+(1-a)\phi_{1}>\phi_1>1,$ which contradicts the
estimate obtained above and therefore $\phi_1 \leq 1$.
Repeating this argument $n$ times, we obtain that all the
functions $\phi_{n+1}$ are bounded, $\phi_{n+1} \leq 1$.

Let us now prove that functions $\phi_{n+1}$ are monotonic. It
is enough to prove that assumption that $\phi_{n}(t)$ is
nonnegative monotonically increasing function leads to
$$
\frac{d}{dt}\left({K}_a [\phi_{n}](t)\right)\geq 0
$$
We have
$$
\frac{d}{dt}\left({K}_a [\phi_{n}](t)\right)
=\int_{0}^{\infty}{\cal K}_{a}^{'}(t,\tau)\phi_{n}(\tau)d\tau
=\int_{0}^{\infty}\left({\cal K}_{1}^{'}(t,\tau)+{\cal K}_{2}^{'}(t,\tau)\right)
\phi_{n}(\tau)d\tau,
$$
where
$$
{{\cal K}_{1}^{'}}(t,\tau)=\frac{1}{2a\sqrt{4 \pi a}}
(\tau-t)e^{-\frac{(t-\tau)^2}{4a}}
$$
and
$$
{{\cal K}_{2}^{'}}(t,\tau)=\frac{1}{2a\sqrt{4 \pi a}}
(\tau+t)e^{-\frac{(t+\tau)^2}{4a}}
$$
Because for all $t,\tau\geq0$ the kernel
${\cal K}_{2}^{'}(t,\tau)\geq 0$ and the function $\phi_{n}(\tau)\geq 0$
and increases it follows that
\begin{multline*}
\frac{d}{dt}\left({K}_a [\phi_{n}](t)\right)
\geq\int_{0}^{\infty}{\cal K}_{1}^{'}(t,\tau)\phi_{n}(\tau)d\tau=
\\
=\int_{0}^{t}{\cal K}_{1}^{'}(t,\tau)\phi_{n}(\tau)d\tau
+\int_{t}^{\infty}{\cal K}_{1}^{'}(t,\tau)\phi_{n}(\tau)d\tau\geq
\\
\geq\int_{0}^{t}{\cal K}_{1}^{'}(t,\tau)\phi_{n}(t)d\tau
+\int_{t}^{\infty}{\cal K}_{1}^{'}(t,\tau)\phi_{n}(t)d\tau
=\phi_{n}(t)\int_{0}^{\infty}{\cal K}_{1}^{'}(t,\tau)d\tau\geq0
\end{multline*}
Using that $\phi_{0}(t)$ is nonnegative monotonically
increasing function, we have
$$
{K}_a [\phi_{0}](t_{0})\leq{K}_a [\phi_{0}](t_{1}),~~  t_{0}\leq t_{1}
$$
or
$$
a\phi_{1}^3(t_0)+(1-a)\phi_{1}(t_0) \leq  a\phi_{1}^3(t_1)+(1-a)\phi_{1}(t_1).
$$
It follows from the inequality above that
$\phi_1 (t_0) \leq \phi_1 (t_1)$ for  $t_0\leq t_1$.
Repeating this argument $n$ times we find that
$\phi_{n+1} (t_0) \leq \phi_{n+1} (t_1)$, for  $t_0\leq t_1$.
We thus have shown that the iterations $\{\phi_{n+1}\}$ are a
sequence of monotonically increasing functions, therefore there
is a limit \cite{Fix}
\begin{equation}
\lim_{n\rightarrow \infty} \phi_{n}(t)=f(t).
\end{equation}
Taking the limit as $n\rightarrow \infty$ in (\ref{itpr})
and using the Lebesque theorem \cite{VS,KM} we obtain equation
\begin{equation}
af^3+(1-a)f-K_af=0,
\end{equation}
where $f \in L_{\infty}[0, \infty]$. Thus the function $f$ is a
solution of equation $(\ref{IE-IP})$. The function $f$  is
bounded, $\phi_n \leq 1$, hence by theorem 2 the function is
continuous. Thus we have proved that our iterative procedure
converges to continuous solution $f$. The function $f$ is
monotonically increasing (since all $\phi_n$ are monotonically
increasing) and is bounded $0 \leq f(t) \leq 1$, so there
exists a limit $\lim\limits_{t \to +\infty} f(t)$. Hence the
function $f$ is bounded from below by initial iteration
$\phi_0$, and from above by one, $\phi_0 \leq f \leq 1$ and
$\lim\limits_{t \to +\infty} \phi_0 =1/2$. Now taking into
account theorem 4, $\lim\limits_{t \to +\infty} f(t)=1$. Thus
the function $f(t)$ is the solution of the boundary value
problem.

We have thus proved that iteration process (\ref{IE-IP})
converges to a continuous solution of the boundary problem
(\ref{krz2}) and (\ref{IE}).
\end{proof}

In conclusion, we have investigated properties of the integral
equation (\ref{eq1}) with infinitely many derivative, we have
constructed an iterative method for solving it and have proved
its convergence.

\section*{Acknowledgments}
The author would like to thank I.Ya. Aref'eva, M.I.Grek and
V.S. Vladimirov for fruitful discussion and valuables remarks.

This paper was supported in part by Dynasty Foundation, the International
Center for Fundamental Physics, the Russian Science Support Foundation,
The Russian Foundation for Basic Research (Grant No. 05-01-00578),
and INTAS (Grant No. 03-51-6346).

\flushleft


\begin{thebibliography}{99}

\bibitem{Sen}
A.~Sen, \textit{Rolling Tachyon},
JHEP 2002, 0204,  048\\
A.~Sen, \textit{Time Evolution in Open String Theory},
JHEP 2002, 0210,  003.

\bibitem{MolZw}
N.~Moeller and B.~Zwiebach,
\textit{Dynamics with Infinitely Many Derivatives and Rolling Tachyons},
JHEP 2002, 0210,  034.

\bibitem{Y}
Yaroslav Volovich,
\textit{Numerical Study of Nonlinear Equations with Infinite Number of Derivatives},
J.Phys.A: Math. Gen. 2003,  \textbf{36},  8685-8701.

\bibitem{AJ}
I.Ya.~Aref'eva, L.V.~Joukovskaya and A.S.~Koshelev,
\textit{Time Evolution in Super\-string Field Theory on non-BPS brane. I. Rolling Tachyon and
Energy--Momentum Conservation}, JHEP (2003) 0309 012.\\
I.Ya.~Aref'eva, Rolling Tachyon in NSSFT, 35-th Ahrenshoop
meeting, Fortschr. Phys., 2003, 51, 652;\\
I.Ya.~Aref'eva, \textit{Nonlocal String Tachyon as a Model for Cosmological Dark Energy},
arXiv:astro-ph/0410443.

\bibitem{MS}
Nicolas Moeller, Martin Schnabl,
\textit{Tachyon condensation in open-closed p-adic string theory},
JHEP 2004, 0401, 011.

\bibitem{Oh}
Kazuki Ohmori,
\textit{Toward Open-Closed String Theoretical Description of Rolling Tachyon},
Phys.Rev. D  2004, 69, 026008.

\bibitem{LY}
L.Joukovskaya, Ya.~Volovich,
\textit{Energy Flow from Open to Closed Strings in a Toy Model of Rolling Tachyon},
math-ph/0308034.

\bibitem{L1}
L.Joukovskaya,
\textit{Energy Conservation for p-Adic and SFT String Equations},
Proceedings of  Steklov Mathematical Institute, 2004,  245, 98.

\bibitem{BFOW}
L.~Brekke, P.G.~Freund, M.~Olson and E.~Witten,
\textit{Nonarchimedean String Dynamics},
Nucl. Phys., 1988, B302, p. 365.

\bibitem{FO}
P.H.~Frampton and Y.~Okada,
\textit{Effective Scalar Field Theory of $p$-Adic String},
Phys. Rev. D, 1988, v. 37, N 10, p.3077-3079.

\bibitem{BF}
L.~Brekke and P.G.O.~Freund,
\textit{$p$-Adic Numbers Physics},
Phys. Rep. (Rev. Sct. Phys. Lett.), 1993, 233, N 1, p.1-66.

\bibitem{VV} V.S. Vladimitov, Ya.I. Volovich,
\textit{On the nonlinear dynamical equation in the $p$-adic string theory},
Theor.Mat.Fiz. \textbf{138} (3), 297 (2004) [Theor.Mat.Fiz.
\textbf{138} (3), 355 (2004)].

\bibitem{VSV} V.S. Vladimirov,
\textit{On the equation of the  $p$-adic open string for the scalar field},
Izv. Math, {\bf 69}, 487 (2005).

\bibitem{Fix}
G.M.~Fikhtengol'ts,
\textit{Founfations of Claculus} [in Russian], Fizmatlit, Moscow (2002),
English transl. prev. ed.; The Fundamentals of Mathematical Physics (Int.
Ser. Monographs Pure Appl. Math., Vols. 72,73), Vols 1,2, Pergamon, Oxford (1965).

\bibitem{VS}
V.S. Vladimirov,
\textit{Equations of Mathematical Physics} [in Russian] (5th ed.),
Nauka, Moscow (1988);
English transl. prev.ed., Marcel Dekker, New York (1971).

\bibitem{KM}
A.N.~Kolmogorov, S.V.~Fomin,
\textit{Elements of the Theory of Foundations and Functional Analysis}
[in Russian], (4th ed.), Nauka, Moscow (1976);
English transl. prev. ed.; Vol.1, Metric and Normal Spaces, Graylock,
Rochester, N.Y. (1957); Vol.2, Measure, The Lebesque Integral, Hilbert Space,
Graylock, Albany, N.Y. (1961).


\end{thebibliography}
\end{document}